\newif\ifFULL
\FULLfalse
\newif\ifLATIN
\LATINfalse

\FULLfalse   
\LATINtrue   

\newif\ifJOURNAL
\JOURNALfalse
\newif\ifCONF
\CONFfalse
\newif\ifarXiv
\arXivfalse
\newif\ifWP
\WPtrue
\newif\ifarXivWP
\arXivWPtrue     

\newif\ifnotFULL   
\notFULLtrue
\ifFULL\notFULLfalse\fi

\newif\ifnotLATIN  
\notLATINtrue
\ifLATIN\notLATINfalse\fi

\ifJOURNAL
  
  \newcommand*{\CTII}{Vovk:2008ECP}
  
  \newcommand*{\CTIV}{Vovk:2012-FS}

  \newcommand*{\CTXI}{Vovk:arXiv1604}

\fi
\ifCONF
  
  \newcommand*{\CTII}{Vovk:2008ECP}
  
  \newcommand*{\CTIV}{Vovk:2012FS-short}

  \newcommand*{\CTXI}{Vovk:arXiv1604}

\fi
\ifarXiv
  
  \newcommand*{\CTII}{Vovk:arXiv0712.1483}
  
  \newcommand*{\CTIV}{Vovk:arXiv0904}

  \newcommand*{\CTXI}{Vovk:arXiv1604}

\fi
\ifWP
  
  \newcommand*{\CTII}{GTP25}
  
  \newcommand*{\CTIV}{GTP28}

  \newcommand*{\CTXI}{GTP43}

\fi
\ifarXivWP   
  
  \renewcommand*{\CTII}{GTP25arXiv}
  
  \renewcommand*{\CTIV}{GTP28arXiv}

  \renewcommand*{\CTXI}{GTP43arXiv}

\fi

\newcommand{\this}{Vovk:2011-LMJ-avoid}  
\newcommand{\zzrelax}[1]{\relax}

\ifnotLATIN
  
\fi
\ifLATIN
  
\fi

\documentclass{article}

\makeatletter

\newif\iftwodates
\twodatesfalse

\renewcommand\maketitle{\begin{titlepage}%
  \let\footnotesize\small
  \let\footnoterule\relax
  \let \footnote \thanks
  \null\vfil
  \vskip 30\p@
  \begin{center}%
    {\LARGE \bf \@title \par}%
    \vskip 3em%
    {\large
     \lineskip .75em%
     \begin{tabular}[t]{c}%
       \@author
     \end{tabular}\par}%
     \vskip 1.5em%
  \end{center}\par
  \vfill
  \begin{center}
    \raisebox{1.5cm}{\includegraphics[width=0.58\textwidth]%
      {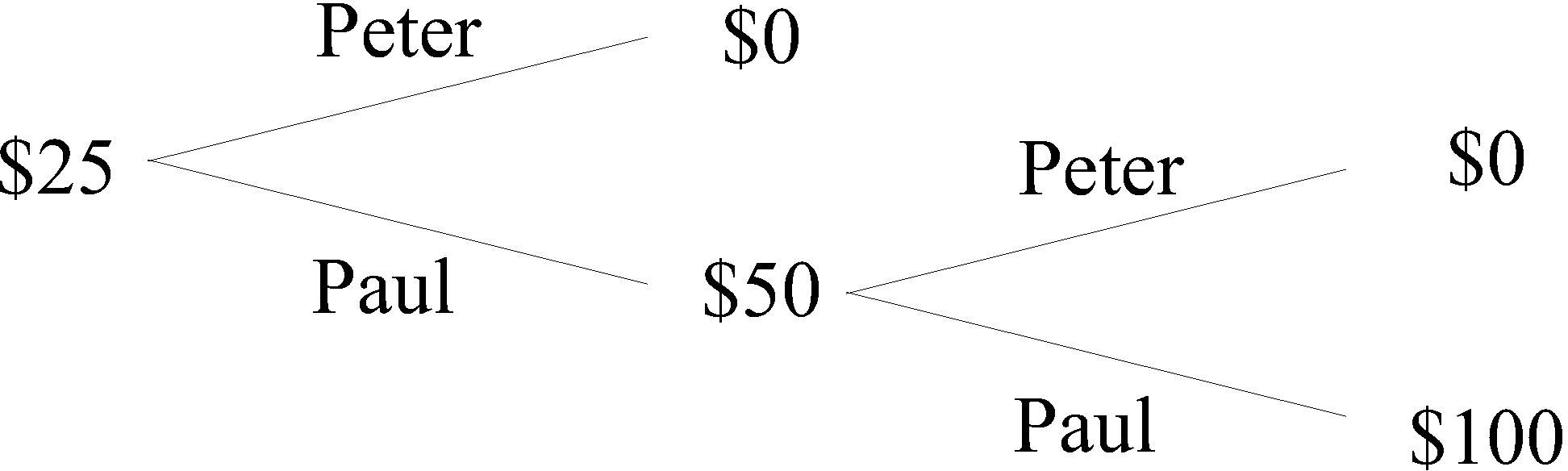}}%
    \hskip 3em%
    \includegraphics[width=0.29\textwidth]%
      {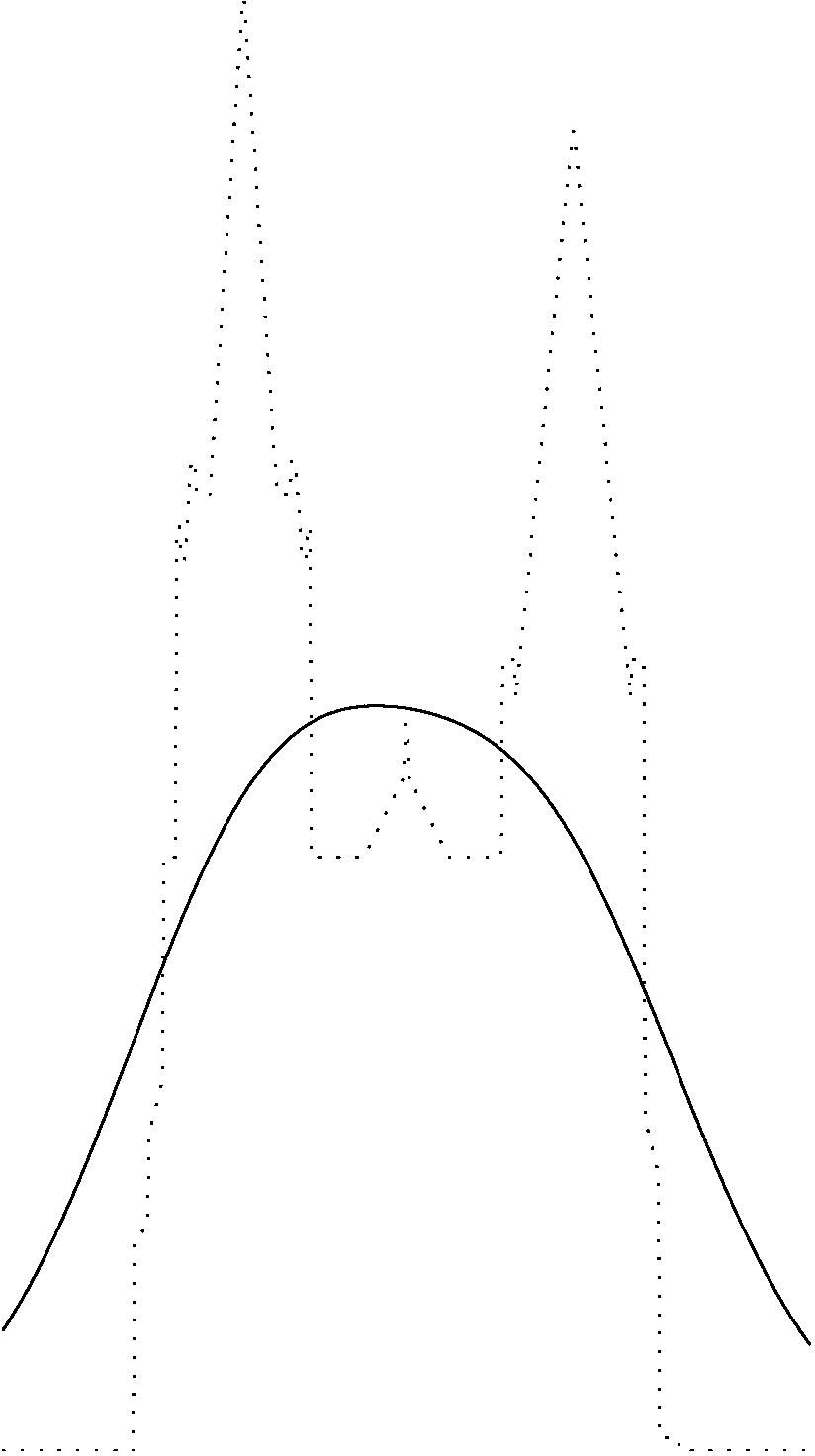}%
  \end{center}
  \@thanks
  \vfill
  \begin{center}
    {\large \bf The Game-Theoretic Probability and Finance Project}
  \end{center}
  \begin{center}
    {\large Working Paper \#\No}
  \end{center}
  \begin{center}
    {\iftwodates\large First posted \firstposted.
    Last revised \@date.\else\large\@date\fi}
  \end{center}
  \begin{center}
    Project web site:\\
    http://www.probabilityandfinance.com
  \end{center}
  \end{titlepage}%
  \setcounter{footnote}{0}%
  \global\let\thanks\relax
  \global\let\maketitle\relax
  \global\let\@thanks\@empty
  \global\let\@author\@empty
  \global\let\@date\@empty
  \global\let\@title\@empty
  \global\let\title\relax
  \global\let\author\relax
  \global\let\date\relax
  \global\let\and\relax
}

\renewenvironment{abstract}{%
  \titlepage
  \null\vfil
  \@beginparpenalty\@lowpenalty
  \begin{center}%
    \Large \bfseries \abstractname
    \@endparpenalty\@M
  \end{center}}%
  {\par\vfill\tableofcontents\endtitlepage}

\renewenvironment{thebibliography}[1]
  {\section*{\refname}%
  \addcontentsline{toc}{section}{\refname}
  \@mkboth{\MakeUppercase\refname}{\MakeUppercase\refname}%
  \list{\@biblabel{\@arabic\c@enumiv}}%
    {\settowidth\labelwidth{\@biblabel{#1}}%
    \leftmargin\labelwidth
    \advance\leftmargin\labelsep
    \@openbib@code
    \usecounter{enumiv}%
    \let\p@enumiv\@empty
    \renewcommand\theenumiv{\@arabic\c@enumiv}}%
    \sloppy
    \clubpenalty4000
    \@clubpenalty \clubpenalty
    \widowpenalty4000%
    \sfcode`\.\@m}
    {\def\@noitemerr
    {\@latex@warning{Empty `thebibliography' environment}}%
  \endlist}

\makeatother

\usepackage{amsmath,amsthm,amsfonts,amssymb,latexsym,graphicx}

\ifFULL
  \usepackage{color}

  \newcommand{\bluebegin}{\begingroup\color{blue}}
  \newcommand{\blueend}{\endgroup}

\fi

\allowdisplaybreaks[4]

\ifnotLATIN
  \input{OT2enc.def}
  
  \usepackage{CJK}
\fi

\newcommand{\st}{\mathrel{|}}		

\renewcommand{\And}{\mathrel{\&}}	

\newcommand{\dd}{\mathrm{d}}		

\newcommand{\K}{\mathcal{K}}		

\newcommand{\BBB}{\mathcal{B}}		
\newcommand{\FFF}{\mathcal{F}}		

\DeclareMathOperator{\III}{\mathbb{I}}		

\newcommand{\bbbp}{\mathbb{P}}		

\DeclareMathOperator{\UpProb}{\overline{\bbbp}}		

\newcommand{\bbbe}{\mathbb{E}}		
\DeclareMathOperator{\Expect}{\bbbe}

\DeclareMathOperator{\var}{v}		
\DeclareMathOperator{\qvar}{w}		
\DeclareMathOperator{\vi}{vi}		
\DeclareMathOperator{\MM}{M}		
\DeclareMathOperator{\DD}{D}		

\newcommand{\bbbr}{\mathbb{R}}		
\newcommand{\bbbz}{\mathbb{Z}}		

\theoremstyle{plain}
\newtheorem{theorem}{Theorem}
\newtheorem{proposition}{Proposition}
\newtheorem{corollary}{Corollary}
\newtheorem{lemma}{Lemma}

\theoremstyle{definition}
\newtheorem{remark}{Remark}
\newtheorem*{question}{Question}

\title{Rough paths in idealized financial markets}
\author{Vladimir Vovk}
\newcommand{\No}{35}
\twodatestrue
\newcommand{\firstposted}{May 4, 2010}

\begin{document}
\maketitle

\begin{abstract}
  This paper considers possible price paths
  of a financial security in an idealized market.
  Its main result is that the variation index of typical price paths
  is at most 2;
  in this sense, typical price paths are not rougher than typical paths of Brownian motion.
  We do not make any stochastic assumptions
  and only assume that the price path is positive and right-continuous.
  The qualification ``typical'' means that there is a trading strategy
  (constructed explicitly in the proof)
  that risks only one monetary unit
  but brings infinite capital when the variation index of the realized price path exceeds 2.
  The paper also reviews some known results for continuous price paths
  and lists several open problems.
\end{abstract}

\section{Introduction}

``Rough paths'' are functions with infinite total variation,
and their roughness is usually measured using the notion of $p$-variation
(\cite{Norvaisa:2006}, p.~102).
Rough paths are ubiquitous in the theory of stochastic processes,
but in recent years they have been actively studied in non-probabilistic settings as well
(see, e.g., \cite{Dudley/Norvaisa:2011} and \cite{Lyons:1998}).
This paper is a contribution to this area of research,
studying price paths of financial securities in idealized markets.
It comes from the tradition of ``game-theoretic probability''
(an approach to probability going back to von Mises and Ville).
No probabilistic assumptions are made about the evolution of security prices
(a non-stochastic notion of probability can be defined,
but this step is optional).
The early work on price paths in game-theoretic probability
relied on using non-standard analysis
(as in \cite{Shafer/Vovk:2001});
this paper follows Takeuchi \emph{et al.}'s recent paper \cite{Takeuchi/etal:2009}
in avoiding non-standard analysis.

We will consider the price path of one financial security over a finite time interval $[0,T]$.
Our key assumption is that the market in our security is efficient,
in the following weak sense:
a prespecified trading strategy risking only 1 monetary unit
will not bring infinite capital at time $T$.
This assumption is not required for our mathematical results,
but is useful in their interpretation
and justifies our terminology:
we say that a property holds for typical price paths
if there is a trading strategy risking only 1 monetary unit
that brings infinite capital at time $T$ whenever the property fails.
Our other assumption is that the interest rate over the time interval $[0,T]$ is 0;
this assumption is easy to relax and is made only for simplicity.

Let $\omega:[0,T]\to[0,\infty)$ be the price path of our financial security;
in this paper we always assume that it is positive
(meaning $\omega\ge0$; this assumption is usually satisfied in real markets).
Section~\ref{sec:cadlag} discusses the case
where $\omega$ is known to be c\`adl\`ag
(i.e., right-continuous and with left limits everywhere).
In this case we can only prove
that the $p$-variation of a typical $\omega$ is finite when $p>2$.
In Section~\ref{sec:continuous} we consider the case
where $\omega$ is known to be continuous.
In this case our understanding is deeper
and we describe briefly some of the much stronger results obtained in \cite{\CTIV}.
A typical result is that the $p$-variation of a non-constant typical $\omega$
is finite when $p>2$ and infinite when $p\le2$;
in particular, the variation index $\vi(\omega)$ of a typical $\omega$
is either $0$ or $2$.
In the last section, Section~\ref{sec:no-borrowing},
we consider markets where borrowing (both borrowing cash and borrowing securities)
is prohibited;
for such markets, the assumption that $\omega$ is continuous loses much of its power.
In Appendix~\ref{app:B} we extend the main result of Section~\ref{sec:cadlag}
to the case where the price path is only assumed to be right-continuous.
In Appendix~\ref{app:C}
we discuss the rationale behind our definitions.

Our approach to rough paths is somewhat different from the standard one,
introduced by Lyons \cite{Lyons:1998}.
Lyons's theory can deal directly only with the rough paths $\omega$ satisfying $\vi(\omega)<2$
(by means of Youngs' theory,
which is described in, e.g., \cite{Dudley/Norvaisa:2011}, Section~2.2).
In order to treat rough paths satisfying $\vi(\omega)\in[n,n+1)$,
where $n=2,3,\ldots$,
we need to postulate the values of the iterated integrals
$X^i_{s,t}:=\int_{s<u_1<\cdots<u_i<t}\dd\omega(u_1)\cdots\dd\omega(u_i)$
for $i=2,\ldots,n$ and $0\le s<t\le T$
(satisfying ``Chen's consistency condition'').
It is not clear how to avoid making an arbitrary choice here.
Our main result (Theorem~\ref{thm:main})
says that only the case $n=2$ is relevant for our idealized markets,
and in this case Lyons's theory is much simpler than in general;
its application in the context of this paper becomes much more feasible,
and would be an interesting direction of further research.

For further discussion of connections with the standard theory of mathematical finance,
including the First and Second Fundamental Theorems of Asset Pricing
and various versions of the no-arbitrage condition,
see \cite{\CTIV}, Sections~1 and~12.

This working paper is based on my talks with the same title
at the Tenth International Vilnius Conference
on Probability and Mathematical Statistics
(section ``Random Processes'', session ``Rough Paths'', 29 June 2010)
and the Third Workshop on Game-theoretic Probability and Related Topics
(21 June 2010).
This is one of the reasons why it not only contains mathematical results
but also discusses the choice of definitions and lists some open problems.
A shorter version of the paper has been published as \cite{\this}.

The words ``positive'' and ``increasing''
are always understood in the wide sense of ``$\ge$'';
the adverb ``strictly'' will be added when needed.
Our notation for logarithms is $\ln$ (natural) and $\log$ (binary).

\section{Volatility of c\`adl\`ag price paths}
\label{sec:cadlag}

Let $\Omega$ be the set $D^+[0,T]$ of all positive c\`adl\`ag functions
$\omega:[0,T]\to[0,\infty)$;
we will call $\Omega$ our \emph{sample space}.
For each $t\in[0,T]$,
$\FFF^{\circ}_t$ is defined to be the smallest $\sigma$-algebra on $\Omega$
that makes all functions
$\omega\mapsto\omega(s)$, $s\in[0,t]$, measurable;
$\FFF_t$ is defined to be the universal completion of $\FFF^{\circ}_t$.
A \emph{process} $S$ is a family of functions
$S_t:\Omega\to[-\infty,\infty]$, $t\in[0,T]$,
each $S_t$ being $\FFF_t$-measurable
(we drop the adjective ``adapted'').
An \emph{event} is an element of the $\sigma$-algebra $\FFF_T$.
Stopping times $\tau:\Omega\to[0,T]\cup\{\infty\}$ w.r.\ to the filtration $(\FFF_t)$
and the corresponding $\sigma$-algebras $\FFF_{\tau}$
are defined as usual;
$\omega(\tau(\omega))$ and $S_{\tau(\omega)}(\omega)$
will be simplified to $\omega(\tau)$ and $S_{\tau}(\omega)$,
respectively
(occasionally,
the argument $\omega$ will be omitted
in other cases as well).

\begin{remark}
  We define $\FFF_t$ to be the universal completion of $\FFF^{\circ}_t$
  in order for the hitting times of closed sets in $\bbbr$ to be stopping times,
  which will be used in the proof of Lemma~\ref{lem:Doob} below.
  Alternatively, we could define $\FFF_t$ as the smaller
  (\cite{Dellacherie/Meyer:1978}, Theorem III.33)
  $\sigma$-algebra generated by the $\FFF^{\circ}_t$-analytic sets:
  see the argument in the proof of Lemma~\ref{lem:Doob}.
\end{remark}

  \begin{remark}
    Another approach would be to define  
    $\FFF_t:=\FFF^{\circ}_{t+}$ (except that $\FFF_T:=\FFF^{\circ}_T$)
    and to use the fact that the hitting times of open sets in $\bbbr$ are stopping times.
    The disadvantage of this definition is that using the filtration $\FFF^{\circ}_{t+}$
    allows ``peeking ahead''.
    It can be argued that in our context
    peeking ahead, just one instant into the future, is tolerable:
    since the price path is right-continuous,
    we can avoid peeking by updating our portfolio an instant later
    rather than now;
    the security price will not change.
    But the counter-argument is that if we are allowed to peek even an instance ahead,
    we can profit greatly even from a single jump in the price process.
    (The first argument works ``forward'', and the second ``backwards''.)
    Therefore, our definition does not use $\FFF^{\circ}_{t+}$.
  \end{remark}

The class of allowed trading strategies is defined in two steps.
A \emph{simple trading strategy} $G$ consists of the following components:
$c\in\bbbr$ (the \emph{initial capital});
an increasing sequence of stopping times
$\tau_1\le\tau_2\le\cdots$;
and, for each $n=1,2,\ldots$,
a bounded $\FFF_{\tau_{n}}$-measurable function $h_n$.
It is required that, for any $\omega\in\Omega$,
only finitely many of $\tau_n(\omega)$ should be finite.
(Intuitively, simple trading strategies
are allowed to trade only finitely often.
Including the initial capital in the trading strategy
is a standard convention in mathematical finance.)
To such $G$ corresponds the \emph{simple capital process}
\begin{equation}\label{eq:simple-capital}
  \K^{G}_t(\omega)
  :=
  c
  +
  \sum_{n=1}^{\infty}
  h_n(\omega)
  \bigl(
    \omega(\tau_{n+1}\wedge t)-\omega(\tau_n\wedge t)
  \bigr),
  \quad
  t\in[0,T]
\end{equation}
(with the zero terms in the sum ignored);
the value $h_n(\omega)$ will be called the \emph{position}
taken at time $\tau_n$,
and $\K^{G}_t(\omega)$ will sometimes be referred to as the capital process of $G$.

A \emph{positive capital process} is any process $S$
that can be represented in the form
\begin{equation}\label{eq:positive-capital}
  S_t(\omega)
  :=
  \sum_{m=1}^{\infty}
  \K^{G_m}_t(\omega),
\end{equation}
where the simple capital processes $\K^{G_m}_t(\omega)$
are required to be positive, for all $t\in[0,T]$ and $\omega\in\Omega$,
and the positive series $\sum_{m=1}^{\infty}c_m$ is required to converge,
where $c_m$ is the initial capital of $G_m$.
The sum (\ref{eq:positive-capital}) is always positive
but allowed to take value $\infty$.
Since $\K^{G_m}_0(\omega)=c_m$ does not depend on $\omega$,
$S_0(\omega)$ also does not depend on $\omega$
and will sometimes be abbreviated to $S_0$.
In our discussions
we will sometimes refer to the sequence $(G_m)_{m=1}^{\infty}$
as a \emph{trading strategy risking $\sum_m c_m$}
and refer to (\ref{eq:positive-capital})
as the \emph{capital process} of this strategy.

\begin{remark}
  The intuition behind the definition of positive capital processes
  is that the initial capital is split into infinitely many accounts
  and the trader runs a separate simple trading strategy on each of these accounts.
  Our definition of simple trading strategies only involves
  the position taken in security, not the cash position.
  The cash position is determined uniquely from the condition
  that the strategy should be self-financing
  (see Section~\ref{sec:no-borrowing}, p.~\pageref{p:cash},
  for further details),
  and in many cases there is no need to mention it explicitly.
  For the explicit connection between our notion of a simple trading strategy
  and the standard definition of a self-financing trading strategy
  (specifying explicitly the cash position, as in, e.g., \cite{Shiryaev:1999}, Section~VII.1a),
  see \cite{\CTIV}, Subsection~2.1.
\end{remark}

\begin{remark}
  Our main result, Theorem~\ref{thm:main}, will continue to hold
  even if the $h_n$ in (\ref{eq:simple-capital}) are required to be constants;
  this will be clear from its proof.
\end{remark}

We say that a set $E\subseteq\Omega$ is \emph{null}
if there is a positive capital process $S$ such that $S_0=1$
and $S_T(\omega)=\infty$ for all $\omega\in E$.
A property of $\omega\in\Omega$ will be said to hold \emph{almost surely} (a.s.),
or \emph{for typical $\omega$},
if the set of $\omega$ where it fails is null.
Intuitively,
we expect such a property to be satisfied in a market
that is efficient at least to some degree.

For each $p\in(0,\infty)$,
the \emph{$p$-variation} $\var_p(f)$ of a function $f:[0,T]\to\bbbr$
is defined as
\begin{equation}\label{eq:var}
  \var_p(f)
  :=
  \sup_\kappa
  \sum_{i=1}^{n}
  \left|
    f(t_{i})-f(t_{i-1})
  \right|^p,
\end{equation}
where $n$ ranges over all strictly positive integers
and $\kappa$ over all \emph{partitions} $0=t_0\le t_1\le\cdots\le t_n=T$
of the interval $[0,T]$.
The \emph{total variation} of a function is the same thing as its 1-variation.
It is obvious that, when $f$ is bounded, there exists a unique number
$\vi(f)\in[0,\infty]$,
called the \emph{variation index} of $f$,
such that $\var_p(f)$ is finite when $p>\vi(f)$
and infinite when $p<\vi(f)$.
It is easy to see that $\vi(f)\notin(0,1)$ when $f$ is continuous,
but in general $\vi(f)$ can take any values in $[0,\infty]$.

\begin{theorem}\label{thm:main}
  For typical $\omega\in D^+[0,T]$,
  \begin{equation}\label{eq:main}
    \vi(\omega)\le2.
  \end{equation}
\end{theorem}

In the case of semimartingales,
the property (\ref{eq:main})
was established by Lepingle (\cite{Lepingle:1976}, Theorem 1(a)).
Intuitively, Theorem \ref{thm:main}
says that price paths cannot be too rough.
In fact, this theorem
can be strengthened to say that there is a trading strategy risking at most 1 monetary unit
whose capital process is $\infty$ at any time $t$
such that the variation index of $\omega$ over $[0,t]$
is greater than 2.
(This remark is also applicable to all other results of this kind in this paper.)
Theorem~\ref{thm:main} will be proved
using Stricker's \cite{Stricker:1979} method
(which is an extension of Bruneau's \cite{Bruneau:1979} method
from continuous to c\`adl\`ag functions).

Let $\MM_a^b(f)$ (resp.\ $\DD_a^b(f)$)
be the number of upcrossings (resp.\ downcrossings) of an open interval $(a,b)$
by a function $f:[0,T]\to\bbbr$ during the time interval $[0,T]$.
For each $h>0$ set
\begin{equation*}
  \MM(f,h)
  :=
  \sum_{k\in\bbbz}
  \MM_{kh}^{(k+1)h} (f),
  \quad
  \DD(f,h)
  :=
  \sum_{k\in\bbbz}
  \DD_{kh}^{(k+1)h} (f).
\end{equation*}
The key ingredient of the proof of Theorem~\ref{thm:main}
is the following game-theoretic version of Doob's upcrossings inequality:
\begin{lemma}\label{lem:Doob}
  Let $0\le a<b$ be real numbers.
  There exists a positive simple capital process $S$
  that starts from $S_0=a$ and satisfies,
  for all $\omega\in\Omega$,
  \begin{equation}\label{eq:Doob}
    S_T(\omega)
    \ge
    (b-a)\MM_a^b(\omega).
  \end{equation}
\end{lemma}
\begin{proof}
  The following standard argument will be easy to formalize.
  A simple trading strategy $G$ leading to $S$
  can be defined as follows.
  The initial capital is $a$.
  At first $G$ takes position $0$.
  When $\omega$ first hits $[0,a]$,
  $G$ takes position $1$ until $\omega$ hits $[b,\infty)$,
  at which point $G$ takes position $0$;
  after $\omega$ hits $[0,a]$,
  $G$ maintains position $1$ until $\omega$ hits $[b,\infty)$,
  at which point $G$ takes position $0$; etc.
  Since $\omega$ is positive, $S$ will also be positive.

  Formally, we define $\tau_1:=\inf\{t\st\omega(t)\in[0,a]\}$
  and, for $n=2,3,\ldots$,
  \begin{equation*} 
    \tau_n
    :=
    \inf\{t\st t>\tau_{n-1} \And \omega(t)\in I_n\},
  \end{equation*}
  where $I_n:=[b,\infty)$ for even $n$
  and $I_n:=[0,a]$ for odd $n$.
  (As usual, the expression $\inf\emptyset$ is interpreted as $\infty$.)
  Since $\omega$ is a right-continuous function and $[0,a]$ and $[b,\infty)$
  are closed sets,
  the infima in the definitions of $\tau_1,\tau_2,\ldots$ are attained.
  Therefore,
  $\omega(\tau_1)\le a$,
  $\omega(\tau_2)\ge b$,
  $\omega(\tau_3)\le a$,
  $\omega(\tau_4)\ge b$, and so on.
  The positions taken by $G$ at the times $\tau_1,\tau_2,\ldots$
  are $h_1:=1$, $h_2:=0$, $h_3:=1$, $h_4:=0$, etc.,
  and the initial capital is $a$.
  Let $n$ be the largest integer such that $\tau_n\le T$
  (with $n:=0$ when $\tau_1=\infty$).
  Now we obtain from (\ref{eq:simple-capital}):
  if $n$ is even,
  \begin{align*}
    S_T(\omega)
    &=
    \K^{G}_T(\omega)\\
    &=
    a
    +
    (\omega(\tau_2)-\omega(\tau_1))
    +
    (\omega(\tau_4)-\omega(\tau_3))
    +\cdots+
    (\omega(\tau_n)-\omega(\tau_{n-1}))\\
    &\ge
    a + (b-a)\MM_a^b(\omega),
  \end{align*}
  and if $n$ is odd,
  \begin{align*}
    S_T(\omega)
    &=
    \K^{G}_T(\omega)\\
    &=
    a
    +
    (\omega(\tau_2)-\omega(\tau_1))
    +
    (\omega(\tau_4)-\omega(\tau_3))
    +\cdots+
    (\omega(\tau_{n-1})-\omega(\tau_{n-2}))\\
    &\quad+
    (\omega(T)-\omega(\tau_{n}))\\
    &\ge
    a + (b-a)\MM_a^b(\omega) 
    +
    (\omega(T)-\omega(\tau_{n}))\\
    &\ge
    a + (b-a)\MM_a^b(\omega) 
    +
    (0-a)
    =
    (b-a)\MM_a^b(\omega);
  \end{align*}
  in both cases, (\ref{eq:Doob}) holds.
  In particular, $S_T(\omega)$ is positive;
  the same argument applied to $t\in[0,T]$ in place of $T$
  shows that $S_t(\omega)$ is positive for all $t\in[0,T]$.

  We have $\tau_n(\omega)<\infty$ for only finitely many $n$
  since $\omega$ is c\`adl\`ag:
  see, e.g., \cite{Dellacherie/Meyer:1978}, Theorem IV.22.
  (This is the only place in this proof
  where we use the assumption that $\omega$ is c\`adl\`ag
  rather than merely right-continuous.)

  It remains to check that each $\tau_n$ is a stopping time;
  we will do so using induction in $n$.
  Let $t\in[0,T]$.
  Since $\omega$ is right-continuous and $[0,a]$ is closed,
  the set $\{\tau_1\le t\}$ is the projection onto $\Omega$
  of the set $A:=\{(s,\omega)\in[0,t]\times\Omega\st\omega(s)\in[0,a]\}$
  (cf.\ \cite{Dellacherie/Meyer:1978}, IV.51(c)).
  Since $A\in\BBB_t\times\FFF^{\circ}_t$,
  where $\BBB_t$ is the Borel $\sigma$-algebra on $[0,t]$
  and $\BBB_t\times\FFF^{\circ}_t$ is the product $\sigma$-algebra,
  the projection $\{\tau_1\le t\}$ is an $\FFF^{\circ}_t$-analytic set
  (according to \cite{Dellacherie/Meyer:1978}, Theorem III.13(3)).
  Therefore, $\{\tau_1\le t\}\in\FFF_t$
  (according to \cite{Dellacherie/Meyer:1978}, Theorem III.33).
  We can see that $\tau_1$ is a stopping time.

  Now let $n\in\{2,3,\ldots\}$ and suppose that $\tau_{n-1}$ is a stopping time.
  Let $t\in[0,T]$.
  Since $\omega$ is right-continuous and $I_n$ is closed,
  the set $\{\tau_n\le t\}$ is the projection onto $\Omega$
  of the set $A:=\{(s,\omega)\in[0,t]\times\Omega\st s>\tau_{n-1}\And\omega(s)\in I_n\}$.
  Since $A\in\BBB_t\times\FFF^{\circ}_t$,
  the same argument as in the previous paragraph shows that $\{\tau_n\le t\}\in\FFF_t$;
  therefore, $\tau_n$ is a stopping time.

  Finally, let us check carefully that the set $\{\tau_n\le t\}$
  is indeed the projection onto $\Omega$
  of $A:=\{(s,\omega)\in[0,t]\times\Omega\st s>\tau_{n-1}\And\omega(s)\in I_n\}$,
  assuming $n>1$
  (the corresponding assertion for $n=1$ is even easier).
  One direction is trivial:
  $s\in[0,t]$, $s>\tau_{n-1}$, and $\omega(s)\in I_n$ immediately implies $\tau_n\le t$.
  In the opposite direction, suppose $\tau_n(\omega)\le t$.
  There is $s\in[0,t]$ and a sequence $t_1\ge t_2\ge\cdots$ such that $\lim_{i\to\infty}t_i=s$
  and, for all $i$,
  $t_i>\tau_{n-1}(\omega)$ and $\omega(t_i)\in I_n$.
  Since $\omega$ is right-continuous and $I_n$ is closed,
  $
    \omega(s)
    =
    \lim_{i\to\infty}
    \omega(t_i)
    \in
    I_n
  $.
  We cannot have $s=\tau_{n-1}$ since $\omega(s)\in I_n$ and $\omega(\tau_{n-1})\notin I_n$.
\end{proof}

In fact, in Proposition~\ref{prop:main} below
we will prove a stronger version of Theorem~\ref{thm:main}.
But to state the stronger version
we will need a generalization of the definition (\ref{eq:var}).
Let $\phi:[0,\infty)\to[0,\infty)$.
For $f:[0,T]\to\bbbr$, we set
\begin{equation*} 
  \var_{\phi}(f)
  :=
  \sup_\kappa
  \sum_{i=1}^{n}
  \phi
  \left(\left|
    f(t_{i})-f(t_{i-1})
  \right|\right),
\end{equation*}
where $\kappa$ ranges over all partitions $0=t_0\le t_1\le\cdots\le t_n=T$,
$n=1,2,\ldots$, of $[0,T]$.
\begin{proposition}\label{prop:main}
  Suppose $\phi:(0,\infty)\to(0,\infty)$ satisfies
  \begin{equation}\label{eq:phi}
    \sup_{0<t\le s\le2t}\frac{\phi(s)}{\phi(t)}<\infty
    \text{\quad and\quad}
    \sum_{j=0}^{\infty}
    2^{2j}
    \phi
    \left(
      2^{-j}
    \right)
    <
    \infty.
  \end{equation}
  Then $\var_{\phi}(\omega)<\infty$ a.s.,
  where $\phi(0)$ is set to $0$.
\end{proposition}

Informally,
the first condition in (\ref{eq:phi}) says that $\phi$ should never increase too fast,
and the second condition says
that $\phi(u)$ should approach $0$ somewhat faster than $u^2$ as $u\to0$.
To obtain Theorem~\ref{thm:main},
set $\phi(u):=u^p$, where $p>2$ is rational,
and notice that the union of countably many null events is always null.
Another simple example of a function $\phi$ satisfying (\ref{eq:phi})
is $\phi(u):=(u/\log^*u)^2$,
where $\log^*u:=1\vee\left|\log u\right|$.
A better example is $\phi(u):=u^2/(\log^*u\log^*\log^*u\cdots)$
(the product is finite if we ignore the factors equal to $1$);
for a proof of (\ref{eq:phi}) for this function,
see \cite{Leung/Cover:1978}, Appendixes B and C
(in this example,
it is essential that $\log$ is binary rather than natural logarithm).
However, even for the last choice of $\phi$,
the inequality $\var_{\phi}(\omega)<\infty$ a.s.\ is still much weaker
than the inequality $\var_{\psi}(\omega)<\infty$ a.s.,
with $\psi$ defined by (\ref{eq:Taylor-function}),
which we can prove assuming $\omega$ continuous
(see Proposition~\ref{prop:Taylor-all} below).
\begin{proof}[Proof of Proposition~\ref{prop:main}]
  Set $w(j):=2^{2j}\phi(2^{-j})$, $j=0,1,\ldots$;
  by (\ref{eq:phi}),
  $\sum_{j=0}^{\infty}w(j)<\infty$.
  Without loss of generality we will assume that $\sum_{j=0}^{\infty}w(j)=1$.

  Let $0=t_0\le t_1\le\cdots\le t_n=T$
  be a partition of the interval $[0,T]$;
  without loss of generality we replace all ``$\le$'' by ``$<$''.
  Fix $\omega\in\Omega$;
  at first we will be mostly interested in the case
  where $\sup_{t\in[0,T]}\omega(t)\le2^L$
  for a given positive integer $L$.
  Split
  $
    \sum_{i=1}^n
    \phi
    \left(\left|
      \omega(t_i) - \omega(t_{i-1})
    \right|\right)
  $
  into two parts:
  \[
    \sum_{i=1}^n
    \phi
    \left(\left|
      \omega(t_i) - \omega(t_{i-1})
    \right|\right)
    =
    \sum_{i\in I_{+}}
    \phi
    \left(
      \omega(t_i) - \omega(t_{i-1})
    \right)
    +
    \sum_{i\in I_{-}}
    \phi
    \left(
      \omega(t_{i-1}) - \omega(t_{i})
    \right),
  \]
  where
  \begin{align*}
    I_{+} &:= \{i \st \omega(t_i) - \omega(t_{i-1}) > 0\},\\
    I_{-} &:= \{i \st \omega(t_i) - \omega(t_{i-1}) < 0\}.
  \end{align*}

  By Lemma~\ref{lem:Doob},
  for each $j=0,1,\ldots$ and each $k\in\{0,\ldots,2^{L+j}-1\}$
  there exists a positive simple capital process $S^{j,k}$
  that starts from $k2^{-j}$ and satisfies
  \begin{equation}\label{eq:proof-1}
    S^{j,k}_T(\omega)
    \ge
    2^{-j}
    \MM_{k2^{-j}}^{(k+1)2^{-j}} (\omega).
  \end{equation}
  Summing $2^{-L-j}S^{j,k}$ over $k=0,\ldots,2^{L+j}-1$
  (in other words, averaging $S^{j,k}$),
  we obtain a positive 
  capital process $S^j$ such that
  \begin{align}
    S_0^j
    &=
    \sum_{k=0}^{2^{L+j}-1}
    k2^{-L-2j}
    \le
    2^{L-1},\notag\\
    S_{T}^j(\omega)
    &\ge
    2^{-L-2j}
    \MM(\omega,2^{-j})
    \text{ when $\sup\omega\le2^L$}.
    \label{eq:proof-2}
  \end{align}
  For each $i\in I_{+}$,
  let $j(i)$ be the smallest positive integer $j$ satisfying
  \begin{equation}\label{eq:j}
    \exists k\in\{0,1,2,\ldots\}:
    \omega(t_{i-1}) \le k2^{-j} \le (k+1)2^{-j} \le \omega(t_i).
  \end{equation}
  Summing $w(j)S^j$ over $j=0,1,\ldots$,
  we obtain a positive capital process $S$ such that
  $
    S_0
    \le
    2^{L-1}
  $
  and, when $\sup\omega\le2^L$,
  \begin{align}
    S_{T}(\omega)
    &\ge
    \sum_{j=0}^{\infty}
    w(j)
    2^{-L-2j}
    \MM(\omega,2^{-j})
    \ge
    \sum_{i\in I_{+}}
    w(j(i))
    2^{-L-2j(i)}\label{eq:to-explain}\\
    &=
    2^{-L}
    \sum_{i\in I_{+}}
    \phi
    \left(
      2^{-j(i)}
    \right) \notag\\
    &\ge
    \delta
    \sum_{i\in I_{+}}
    \phi
    \left(
      \omega(t_i)-\omega(t_{i-1})
    \right),\label{eq:end}
  \end{align}
  where $\delta>0$ depends only on $L$
  and the supremum in (\ref{eq:phi}).
  The second inequality in (\ref{eq:to-explain})
  follows from the fact that to each $i\in I_+$
  corresponds an upcrossing of an interval of the form
  $(k2^{-j(i)},(k+1)2^{-j(i)})$.

  An inequality analogous to the inequality
  between the second and the last terms of the chain (\ref{eq:to-explain})--(\ref{eq:end})
  can be proved for downcrossings instead of upcrossings,
  $I_-$ instead of $I_+$,
  and $\omega(t_{i-1})$ and $\omega(t_i)$ swapped around.
  Using this inequality (in the third ``$\ge$'' below) gives,
  when $\sup\omega\le2^L$,
  \begin{align}
    S_{T}(\omega)
    &\ge
    \sum_{j=0}^{\infty}
    w(j)
    2^{-L-2j}
    \MM(\omega,2^{-j})
    \label{eq:chain-start}\\
    &\ge
    \sum_{j=0}^{\infty}
    w(j)
    2^{-L-2j}
    \left(
      \DD(\omega,2^{-j}) - 2^{L+j}
    \right)
    \notag\\
    &\ge
    \delta
    \sum_{i\in I_{-}}
    \phi
    \left(
      \omega(t_{i-1})-\omega(t_{i})
    \right)
    -
    \sum_{j=0}^{\infty}
    w(j)2^{-j}
    \notag\\
    &\ge
    \delta
    \sum_{i\in I_{-}}
    \phi
    \left(
      \omega(t_{i-1})-\omega(t_{i})
    \right)
    -
    1.
    \label{eq:chain-end}
  \end{align}
  Averaging the two lower bounds for $S_T(\omega)$,
  we obtain, when $\sup\omega\le2^L$,
  \[
    S_T(\omega)
    \ge
    \frac{\delta}{2}
    \sum_{i=1}^n
    \phi
    \left(\left|
      \omega(t_{i})-\omega(t_{i-1})
    \right|\right)
    -
    \frac{1}{2}.
  \]
  Taking supremum over all partitions gives
  \begin{equation}\label{eq:inequality}
    \sup\omega\le2^L
    \Longrightarrow
    S_T(\omega)
    \ge
    \frac{\delta}{2}
    \var_{\phi}(\omega)
    -
    \frac{1}{2}.
  \end{equation}
  We can see that the event that $\sup\omega\le2^L$
  and $\var_{\phi}(\omega)=\infty$ is null.
  Since the union of countably many null events is always null,
  the event that $\var_{\phi}(\omega)=\infty$ is also null.
\end{proof}

The case of c\`adl\`ag price paths considered in this section
is very different from the case of continuous price paths
that we take up in the following section.
Proposition~\ref{prop:vi-continuous} will show that,
in the latter case,
$\vi(\omega)\in\{0,2\}$ a.s.
In the former case,
no c\`adl\`ag price path that is bounded away from zero
and has finite total variation can belong to a null event,
as the following proposition will show.

The \emph{upper probability} of a set $E\subseteq\Omega$
is defined as
\begin{equation}\label{eq:upper-probability}
  \UpProb(E)
  :=
  \inf
  \bigl\{
    S_0
    \bigm|
    \forall\omega\in\Omega:
    S_T(\omega)
    \ge
    \III_E(\omega)
  \bigr\},
\end{equation}
where $S$ ranges over the positive capital processes
and $\III_E$ stands for the indicator of $E$.
In this section we will be interested only in one-element sets $E$.
We write $\var(f)$ meaning $\var_1(f)$.
\begin{proposition}\label{prop:bounded-variation}
  For any $\omega\in\Omega$,
  \begin{equation}\label{eq:bounded-variation}
    \UpProb(\{\omega\})
    =
    \sqrt{\frac{\omega(0)}{\omega(T)}e^{-\var(\ln\omega)}}.
  \end{equation}
\end{proposition}
\begin{proof}
  Fix $\omega\in\Omega$.
  Let $S$ be any positive capital process.
  Represent it in the form (\ref{eq:positive-capital}).
  It suffices to prove that none of the component strategies $G_m$
  can increase its initial capital $c_m$
  by more than a factor of
  \[
    \sqrt{\frac{\omega(T)}{\omega(0)}e^{\var(\ln\omega)}}
  \]
  and that this factor itself is attainable,
  at least in the limit.
  Fix an $m$ and let $c=c_m$, $\tau_1,\tau_2,\ldots$, and $h_1,h_2,\ldots$
  be the component initial capital, stopping times, and positions of $G_m$.
  It is clear that all $h_n$ must be positive
  in order for $\K:=\K^{G_m}$ to be positive:
  upward price movements are unbounded.
  Downward price movements right after $\tau_n$
  can be as large as $\omega(\tau_n)$,
  which implies that
  \begin{equation}\label{eq:super-prudent}
    0 \le h_n \le \K_{\tau_n} / \omega(\tau_n)
  \end{equation}
  (this condition will be further discussed and justified in Section~\ref{sec:no-borrowing}).
  This gives, according to (\ref{eq:simple-capital}),
  \[
    \K_{\tau_{n+1}}
    =
    \K_{\tau_n}
    +
    h_n
    \left(
      \omega(\tau_{n+1})
      -
      \omega(\tau_n)
    \right)
    \le
    \left(
      1
      \vee
      \frac
      {\omega(\tau_{n+1})}
      {\omega(\tau_n)}
    \right)
    \K_{\tau_n}.
  \]
  The last ``$\le$'' becomes ``$=$'' when
  \[
    h_n
    :=
    \begin{cases}
      \K_{\tau_n} / \omega(\tau_n) & \text{if $\omega(\tau_{n+1})>\omega(\tau_n)$}\\
      0 & \text{otherwise}.
    \end{cases}
  \]
  We can see that no positive simple capital process increases its initial capital
  by more than a factor of $e^{\var^+(\ln\omega)}$,
  where $\var^+(f)$ is defined by the following modification of (\ref{eq:var}):
  \begin{equation*}
    \var^+(f)
    :=
    \sup_\kappa
    \sum_{i=1}^{n}
    \left(
      f(t_{i})-f(t_{i-1})
    \right)^+;
  \end{equation*}
  as usual, $u^+$ and $u^-$ are defined to be $0\vee u$ and $0\vee(-u)$, respectively.
  On the other hand,
  for each $\epsilon>0$,
  there is a positive simple capital process that increases its initial capital
  by a factor of at least $(1-\epsilon)e^{\var^+(\ln\omega)}$.
  We can see that
  \begin{equation}\label{eq:almost}
    \UpProb(\{\omega\})
    =
    e^{-\var^+(\ln\omega)}.
  \end{equation}
  If we define
  \begin{equation*}
    \var^-(f)
    :=
    \sup_\kappa
    \sum_{i=1}^{n}
    \left(
      f(t_{i})-f(t_{i-1})
    \right)^-,
  \end{equation*}
  we can further see that $\var(f)=\var^+(f)+\var^-(f)$
  and $f(T)-f(0)=\var^+(f)-\var^-(f)$;
  the last two equalities imply $\var^+(f)=(\var(f)+f(T)-f(0))/2$.
  In combination with (\ref{eq:almost}),
  this gives (\ref{eq:bounded-variation}).
\end{proof}

\section{Volatility of continuous price paths}
\label{sec:continuous}

In this section we consider a new sample space:
$\Omega$ is now the set $C^+[0,T]$ of all positive continuous functions
$\omega:[0,T]\to[0,\infty)$.
Intuitively,
this is the set of all possible price paths of our security.
For each $t\in[0,T]$,
the $\sigma$-algebra $\FFF'_t$ on $C^+[0,T]$ is the trace of $\FFF_t$ on $C^+[0,T]$
(i.e., $\FFF'_t$ consists of the sets $E\cap C^+[0,T]$ with $E\in\FFF_t$);
we will omit the prime in $\FFF'_t$.
A \emph{process} $S$ is a family of functions
$S_t:C^+[0,T]\to[-\infty,\infty]$, $t\in[0,T]$,
such that each $S_t$ is $\FFF_t$-measurable.
A \emph{simple capital process} is defined to be
the restriction of a simple capital process in the old sense to $C^+[0,T]$
(i.e., $S'$ is called a simple capital process
if there is a simple capital process $S$ in the old sense
such that, for each $t\in[0,T]$, $S'_t=S_t|_{C^+[0,T]}$).
\emph{Positive capital processes} are capital processes $S$
that can be represented in the form (\ref{eq:positive-capital}),
where the simple capital processes $\K^{G_m}_t(\omega)$
are required to be positive,
for all $t\in[0,T]$ and $\omega\in C^+[0,T]$,
and the positive series $\sum_{m=1}^{\infty}c_m$ is required to converge,
where $c_m$ is the initial capital of $G_m$.
(The notion of simple trading strategies does not change,
but we are only interested in their behaviour on $\omega\in C^+[0,T]$.)
An \emph{event} is an element of the $\sigma$-algebra $\FFF_T$ on $C^+[0,T]$.
The definition of a null event is the same as before
(but using the new notion of a positive capital process),
and the adjective ``typical'' will again be used to refer to the complements
of null events.

\begin{remark}
  The definitions used in \cite{\CTIV} are slightly different,
  but all proofs there also work under our current definitions.
  Under the assumption of continuity of $\omega$,
  the requirement that $\omega$ should be positive is superfluous,
  and is never made in \cite{\CTIV}.
\end{remark}

The following elaboration of Theorem \ref{thm:main} for continuous price paths
was established in \cite{\CTII} using direct arguments
(relying on the result in \cite{Bruneau:1979} mentioned earlier
for the inequality $\vi(\omega)\le2$
and a standard argument,
going back to \cite{Kolmogorov:1929LLN}
and used in the context of mathematical finance in \cite{Shiryaev:1999},
Example~3 on p.~658,
for the inequality $\vi(\omega)\ge2$ for non-constant $\omega$).

\begin{proposition}[\cite{\CTII}, Theorem~1]\label{prop:vi-continuous}
  For typical $\omega\in C^+[0,T]$,
  \begin{equation*} 
    \vi(\omega)=2 \text{ or $\omega$ is constant}.
  \end{equation*}
\end{proposition}

This proposition is similar to the well-known property of continuous semimartingales
(Lepingle \cite{Lepingle:1976}, Theorem 1(a) and Proposition 3(b)).
Related results in mathematical finance usually make strong stochastic assumptions
(such as those in \cite{Rogers:1997local}).
A probability-free result related to the inequality $\vi(\omega)\ge2$
(for typical non-constant $\omega$)
was established by Salopek \cite{Salopek:1998} (p.~228),
who proved that the trader can start from $0$ and end up with a strictly positive capital
in a market with two securities whose price paths $\omega_1$ and $\omega_2$
are strictly positive, continuous, and satisfy $\vi(\omega_1)<2$, $\vi(\omega_2)<2$,
$\omega_1(0)=\omega_2(0)=1$, and $\omega_1(T)\ne\omega_2(T)$.
However, Salopek's definition of a capital process only works
under the assumption that all securities in the market
have price paths $\omega$ satisfying $\vi(\omega)<2$.
The proof of Salopek's result was simplified in \cite{Norvaisa:2000-full}
(using the argument from \cite{Shiryaev:1999} mentioned earlier).

The paper \cite{\CTIV} establishes connections between continuous price paths
and Brownian motion,
which in combination with Taylor's \cite{Taylor:1972} results
greatly refine Proposition~\ref{prop:vi-continuous}.
Let $\psi:[0,\infty)\to[0,\infty)$ be Taylor's \cite{Taylor:1972} function
\begin{equation}\label{eq:Taylor-function}
  \psi(u)
  :=
  \frac{u^2}{2\ln^*\ln^*u},
\end{equation}
with $\psi(0):=0$ and $\ln^*u:=1\vee\left|\ln u\right|$.

\begin{proposition}[\cite{\CTIV}, Corollary~5]\label{prop:Taylor-all}
  For typical $\omega\in C^+[0,T]$,
  \begin{equation*} 
    \var_{\psi}(\omega)<\infty.
  \end{equation*}
  Suppose $\phi:[0,\infty)\to[0,\infty)$ is such that $\psi(u)=o(\phi(u))$ as $u\to0$.
  For typical $\omega\in C^+[0,T]$,
  \begin{equation*} 
    \var_{\phi}(\omega)=\infty \text{ or $\omega$ is constant}.
  \end{equation*}
\end{proposition}


\begin{question}
  Can Proposition~\ref{prop:Taylor-all}
  be partially extended to positive c\`adl\`ag functions
  to say that $\var_{\psi}(\omega)<\infty$ for typical $\omega\in D^+[0,T]$?
\end{question}

Proposition~\ref{prop:Taylor-all} refines Proposition~\ref{prop:vi-continuous},
but is further strengthened by the next result, Proposition~\ref{prop:Taylor}.
The following quantity was introduced by Taylor \cite{Taylor:1972}:
for $f:[0,T]\to\bbbr$, set
\begin{equation*} 
  \qvar(f)
  :=
  \lim_{\delta\to0}
  \sup_{\kappa\in K_{\delta}}
  \sum_{i=1}^{n}
  \psi
  \left(
    \left|
      f(t_i)
      -
      f(t_{i-1})
    \right|
  \right),
\end{equation*}
where $K_{\delta}$ is the set of all partitions $0=t_0\le\cdots\le t_{n}=T$
of $[0,T]$ whose mesh is less than $\delta$:
$\max_i(t_i-t_{i-1})<\delta$.
Notice that $\qvar(\omega)\le\var_{\psi}(\omega)$.

\begin{proposition}[\cite{\CTIV}, Corollary~6]\label{prop:Taylor}
  For typical $\omega\in C^+[0,T]$,
  \begin{equation*} 
    \qvar(\omega)\in(0,\infty) \text{ or $\omega$ is constant}.
  \end{equation*}
\end{proposition}

\section{The case of no borrowing}
\label{sec:no-borrowing}

The definitions in this section are applicable
both to the framework of Section~\ref{sec:cadlag}
(where the sample space is the set $\Omega:=D^{+}[0,T]$
of all positive c\`adl\`ag functions on $[0,T]$)
and to the framework of Section~\ref{sec:continuous}
(where the sample space is the set $\Omega:=C^+[0,T]$
of all positive continuous functions on $[0,T]$).
In this paper, we only use positive capital processes $S_t$.
However, even positive capital processes
may involve borrowing cash or security:
at each time, $S_t$ is the price of a portfolio
containing some amounts of security and cash;
the total value of the portfolio is positive
but nothing prevents either of its components to be strictly negative.
In this section we consider a market
where the trader is allowed to borrow neither cash nor security
(borrowing security is essentially the same thing as short-selling in this context).
Such markets have been considered by, e.g.,
Cover \cite{Cover:1991} and Koolen and de Rooij \cite{Koolen/deRooij:2010local}.

Let $G$ be a simple trading strategy.
As before, the components of $G$ will be denoted
$c$ (the initial capital),
$\tau_n$ (the stopping times), and $h_n$ (the positions),
and we imagine a trader who follows $G$.
For $t\in[0,T]$ and $\omega\in D^+[0,T]$ or $\omega\in C^+[0,T]$ (as appropriate),
set $h_t(\omega):=h_n(\omega)$,
where $n$ is the unique number satisfying $t\in(\tau_n,\tau_{n+1}]$
(with $h_t(\omega):=0$ if $t\le\tau_1(\omega)$);
intuitively,
$h_t$ is the trader's position at time $t$.
The amount of \label{p:cash}\emph{cash} in the trader's portfolio at time $t$
is defined to be
$\K^{G}_t(\omega) - h_t(\omega) \omega(t)$.
Let us say that the trading strategy $G$ is \emph{borrowing-free}
if, for all $\omega$ and $t$,
we have $h_t(\omega) \ge 0$
(the condition of \emph{no borrowing security})
and $\K^{G}_t(\omega) - h_t(\omega) \omega(t) \ge 0$
(the condition of \emph{no borrowing cash}).
(Remember that being borrowing-free
is a completely different requirement from being self-financing:
all trading strategies considered in this paper are self-financing.)

It is easy to see that $G$ is borrowing-free
if and only if $c\ge0$ and (\ref{eq:super-prudent})
(with $\K$ understood to be $\K^G$)
is satisfied for all $n\in\{1,2,\ldots\}$.
Indeed, suppose the latter condition is satisfied.
If $t\in[0,\tau_1(\omega)]$,
\[
  \K^G_t(\omega)
  -
  h_t(\omega)\omega(t)
  =
  c
  \ge
  0.
\]
And if $t\in(\tau_n(\omega),\tau_{n+1}(\omega)]$,
\begin{align*}
  \K^G_t(\omega)
  -
  h_t(\omega)\omega(t)
  &=
  \K^G_{\tau_n}(\omega)
  +
  h_n(\omega) (\omega(t)-\omega(\tau_n))
  -
  h_n(\omega)\omega(t) \\
  &=
  \K^G_{\tau_n}(\omega)
  -
  h_n(\omega)\omega(\tau_n)
  \ge
  0.
\end{align*}

In the framework of Section~\ref{sec:cadlag},
where the sample space is $D^+[0,T]$,
all trading strategies $G$ for which $\K^G$ is positive are automatically borrowing-free
(we already used this fact in the proof of Proposition~\ref{prop:bounded-variation}).
Indeed, let $\K^G$ be positive.
If the condition of no borrowing security is violated and $h_t(\omega)<0$,
we can make $\K_t^{G}(\omega)<0$ by modifying $\omega$ over $[t,T]$
and making $\omega(t)$ sufficiently large.
(Intuitively, borrowing security is risky when its price can jump
since there is no upper limit on the price.)
If the condition of no borrowing cash is violated
and $\K^{G}_t(\omega) - h_t(\omega) \omega(t) < 0$,
we can make $\K_t^{G}(\omega)<0$ by modifying $\omega$ over $[t,T]$
and setting $\omega(t):=0$.
(Intuitively, borrowing cash is risky when the security's price can jump
since the price can drop to zero at any time.)
We will see shortly that in the framework of Section~\ref{sec:continuous},
where the sample space is $C^+[0,T]$,
the condition that $G$ should be borrowing-free makes a big difference.

By a \emph{borrowing-free capital process} we will mean a process $S$
that can be represented in the form (\ref{eq:positive-capital})
where all trading strategies $G_m$ are required to be borrowing-free
and the positive series $\sum_{m=1}^{\infty}c_m$ is required to converge.
This definition is applicable to the frameworks
of both Section~\ref{sec:cadlag} and Section~\ref{sec:continuous}.

Let $E$ be a set of positive continuous functions on $[0,T]$.
Since $E\subseteq D^+[0,T]$ and $E\subseteq C^+[0,T]$,
\emph{a priori}
there are at least four natural definitions of the upper probability $\UpProb(E)$:
\begin{itemize}
\item
  $\UpProb_1(E)$ is the upper probability (\ref{eq:upper-probability})
  with $S$ ranging over the positive capital processes
  defined on the space $C^+[0,T]$ of all positive continuous functions on $[0,T]$;
\item
  $\UpProb_2(E)$ is the upper probability (\ref{eq:upper-probability}),
  exactly as it is defined there;
  namely, $S$ ranges over the positive capital processes
  defined on the space $D^{+}[0,T]$ of all positive c\`adl\`ag functions
  on $[0,T]$;
\item
  $\UpProb_3(E)$ is the upper probability (\ref{eq:upper-probability})
  with $S$ ranging over the borrowing-free capital processes
  defined on $C^+[0,T]$.
\item
  $\UpProb_4(E)$ is the upper probability (\ref{eq:upper-probability})
  with $S$ ranging over the borrowing-free capital processes
  defined on $D^{+}[0,T]$.
\end{itemize}
In fact, most of these definitions are equivalent:
\begin{proposition}
  For any set $E\subseteq C^+[0,T]$ of positive continuous functions on $[0,T]$,
  \[
    \UpProb_1(E)
    \le
    \UpProb_2(E)
    =
    \UpProb_3(E)
    =
    \UpProb_4(E).
  \]
  There exists a set $E$ of positive continuous functions on $[0,T]$ such that
  \[
    \UpProb_1(E)
    =0<1=
    \UpProb_2(E)
    =
    \UpProb_3(E)
    =
    \UpProb_4(E).
  \]
\end{proposition}
\begin{proof}
  The equality $\UpProb_2(E)=\UpProb_4(E)$ has already been demonstrated,
  and the equality $\UpProb_3(E)=\UpProb_4(E)$ is not difficult to prove.
  Therefore, $\UpProb_2(E)=\UpProb_3(E)=\UpProb_4(E)$.
  Now let $E$ be the set of all $\omega\in C^+[0,T]$
  such that $\vi(\omega)\in(0,2)$.
  According to Proposition~\ref{prop:vi-continuous},
  $\UpProb_1(E)=0$.
  And according to Proposition~\ref{prop:bounded-variation},
  $\UpProb_2(E)=1$:
  there are even individual elements $\omega\in E$
  for which $\UpProb_2(\{\omega\})$ is arbitrarily close to~$1$
  (such as $\omega(t)=1+\epsilon t$
  for sufficiently small $\left|\epsilon\right|\ne0$).
\end{proof}

\subsection*{Acknowledgments}

Akimichi Takemura's, George Lowther's, and an anonymous referee's advice
is very much appreciated.
I am grateful to Rimas Norvai\v{s}a for inviting me
to give the talk at the Vilnius conference
and for useful discussions.
This work was partially supported by EPSRC (grant EP/F002998/1).

\appendix

\section{The case of finite $p$-variation, $p>2$}

Let $p>2$.
Theorem \ref{thm:main} says that the trader can become infinitely rich
when $\var_p(\omega)=\infty$.
This appendix treats the case where $\var_p(\omega)$ is merely large,
not infinitely large.
We are now in the framework of Section \ref{sec:cadlag}:
the sample space is $\Omega:=D^+[0,T]$.

\begin{proposition}\label{prop:explicit}
  Let $p=2+\epsilon>2$ and let $\delta>0$.
  There is a positive capital process $S$ such that $S_0=1$
  and, for all $\omega\in\Omega$,
  \begin{equation}\label{eq:explicit}
    S_T(\omega)
    >
    (1-2^{-\epsilon})
    (1-2^{-\delta})
    2^{-6-\epsilon-\delta}
    \frac{\var_{2+\epsilon}(\omega)}{(1\vee\sup\omega)^{2+\epsilon+\delta}}
    -
    \frac14.
  \end{equation}
\end{proposition}

\begin{proof}
  In this proof we will see what the argument used in the proof of Theorem \ref{thm:main}
  gives in the case of a finite $\var_p(\omega)$.
  It will be convenient to modify the function $j(i)$ used in that argument,
  making it dependent on the given upper bound $2^L$ on $\omega$.
  For $L\in\{0,1,2,\ldots\}$,
  define $j_L(i)$ to be the smallest integer $j\ge2-L$ satisfying (\ref{eq:j}).
  This definition ensures that $2^{-j_L(i)}\ge\frac14(\omega(t_i)-\omega(t_{i-1}))$
  when $\sup\omega\le2^L$.

  Fix temporarily $L\in\{0,1,2,\ldots\}$.
  Now we set
  $w(j):=(1-2^{-\epsilon})2^{\epsilon(2-L)}2^{-\epsilon j}$,
  $j=2-L,3-L,\ldots$;
  $(1-2^{-\epsilon})2^{\epsilon(2-L)}$ is the normalizing constant ensuring
  $\sum_{j=2-L}^{\infty}w(j)=1$.
  Using the inequality between the two extreme terms in (\ref{eq:to-explain})
  (with the lower limit of summation $j=2-L$ instead of $j=0$)
  and setting $S^{(L)}:=2^{1-L}S$,
  we obtain a positive capital process satisfying $S^{(L)}_0\le1$ and
  \begin{align*}
    S^{(L)}_{T}(\omega)
    &=
    2^{1-L}
    S_{T}(\omega)
    \ge
    2^{1-2L}
    \sum_{i\in I_{+}}
    w(j_L(i))
    (2^{-j_L(i)})^2\\
    &=
    2^{1-2L}
    (1-2^{-\epsilon})
    2^{\epsilon(2-L)}
    \sum_{i\in I_{+}}
    (2^{-j_L(i)})^{2+\epsilon}\\
    &\ge
    2^{1-2L}
    (1-2^{-\epsilon})
    2^{\epsilon(2-L)}
    4^{-2-\epsilon}
    \sum_{i\in I_{+}}
    (\omega(t_i)-\omega(t_{i-1}))^{2+\epsilon}
  \end{align*}
  instead of (\ref{eq:to-explain})--(\ref{eq:end}).
  And instead of (\ref{eq:inequality}) we now obtain
  \begin{align*}
    \sup\omega\le2^L
    \Longrightarrow
    S^{(L)}_T(\omega)
    &\ge
    2^{-2L}
    (1-2^{-\epsilon})
    2^{\epsilon(2-L)}
    4^{-2-\epsilon}
    \var_{2+\epsilon}(\omega)\\
    &\quad-
    2^{-L}
    \sum_{j=2-L}^{\infty}
    w(j)2^{-j}\\
    &>
    (1-2^{-\epsilon})
    2^{-2L-\epsilon L}
    4^{-2}
    \var_{2+\epsilon}(\omega)
    -
    \frac14.
  \end{align*}

  Set $S:=\sum_{L=0}^{\infty}(1-2^{-\delta})2^{-\delta L}S^{(L)}$
  (recycling the notation $S$);
  $1-2^{-\delta}$ is the normalizing constant
  ensuring that the weights
  $(1-2^{-\delta})2^{-\delta L}$
  sum to $1$.
  For any $\omega$ and any upper bound $2^L\ge\sup\omega$,
  with $L\in\{0,1,2,\ldots\}$,
  we will have
  \[
    S_T(\omega)
    \ge
    (1-2^{-\delta}) 2^{-\delta L}
    S^{(L)}_T(\omega)
    >
    (1-2^{-\epsilon}) (1-2^{-\delta})
    4^{-2}
    (2^L)^{-2-\epsilon-\delta}
    \var_{2+\epsilon}(\omega)
    -
    \frac14.
  \]
  Taking the $L$ satisfying $1\vee\sup\omega\le2^L<2(1\vee\sup\omega)$,
  we obtain
  \begin{equation*}
    S_T(\omega)
    >
    (1-2^{-\epsilon})
    (1-2^{-\delta})
    4^{-2}
    2^{-2-\epsilon-\delta}
    (1\vee\sup\omega)^{-2-\epsilon-\delta}
    \var_{2+\epsilon}(\omega)
    -
    \frac14,
  \end{equation*}
  which is equivalent to (\ref{eq:explicit}).
\end{proof}

Proposition \ref{prop:explicit} is mainly motivated by the case of discrete time.
Suppose the trader is allowed to change his positions in $\omega$ only at times
$0,T/N,2T/N,\ldots,T$ for a strictly positive integer $N$.
This restriction is equivalent to replacing $\omega$ by $\omega_N\in D^+[0,T]$
defined by
\[
  \omega_N(t)
  :=
  \omega
  \left(
    \frac{T}{N}
    \left\lfloor
      \frac{N}{T}
      t
    \right\rfloor
  \right),
  \quad
  t\in[0,T].
\]
The discrete-time version of Proposition~\ref{prop:explicit}
(which is weaker than Proposition~\ref{prop:explicit} itself)
says that there is a positive capital process $S$
such that $S_0=1$ and (\ref{eq:explicit})
holds for all elements of $\Omega$ of the form $\omega_N$.

\section{Right-continuous price paths}
\label{app:B}

In this appendix we will relax the assumption that the price path $\omega$ is c\`adl\`ag,
and will consider the sample space $\Omega:=R^+[0,T]$ consisting
of all positive right-continuous functions $\omega:[0,T]\to[0,\infty)$.
The definitions of the $\sigma$-algebras $\FFF_t$, processes, events,
simple capital processes, positive capital processes,
and the qualification ``almost surely''
stay literally the same as in Section~\ref{sec:cadlag}.
We will check that Theorem~\ref{thm:main} will continue to hold
in this less restrictive framework.
But first we state the following simple version
of Theorem VI.3(2) in \cite{Dellacherie/Meyer:1982}.

\begin{corollary}\label{cor:wild}
  Almost surely, the price path $\omega\in R^+[0,T]$ is c\`adl\`ag.
\end{corollary}

\begin{proof}
  %
  We start by noticing that a typical $\omega\in R^+[0,T]$ is bounded above.
  Indeed, for $m=1,2,\ldots$,
  let $G_m$ be the simple trading strategy with initial capital 1,
  stopping times $\tau_1:=0$,
  $\tau_2:=\inf\{t\st\omega(t)\ge2^m\}$,
  $\tau_3=\tau_4=\cdots:=\infty$,
  and positions $h_1:=1/\omega(0)$,
  $h_2=h_3=\cdots:=0$
  (if $\omega(0)=0$, set $h_1:=1$).
  The positive capital process $\sum_{m=1}^{\infty}2^{-m}\K^{G_m}$
  has initial capital 1 and final capital $\infty$ on unbounded $\omega$.

  Now it suffices to prove that the number of upcrossings of any open interval $(a,b)$
  with rational endpoints is finite almost surely
  (\cite{Dellacherie/Meyer:1978}, Theorem IV.22).
  Fix a set of weights $w(a,b)>0$
  such that $\sum_{(a,b)}aw(a,b)<\infty$ and $\sum_{(a,b)}w(a,b)=1$,
  $(a,b)$ ranging over the open intervals with rational endpoints $a\ge0$ and $b>a$.
  For each $(a,b)$,
  let $S^{(a,b)}$ be the simple capital process $S$ from the proof Lemma~\ref{lem:Doob}
  modified as follows:
  to ensure that $\tau_n<\infty$ for only finitely many $n$,
  we stop trading when $S^{(a,b)}$ reaches the value $1/w(a,b)$.
  The positive capital process $\sum_{(a,b)}w(a,b)S^{(a,b)}$
  has a finite initial value and the infinite final value
  whenever the number of upcrossings of some open interval $(a,b)$ is infinite:
  indeed, if the interval $(a,b)$ is crossed infinitely often,
  any of its subintervals will be crossed infinitely often as well.
\end{proof}

Corollary~\ref{cor:wild} does not mean that the results
that we have proved above for c\`adl\`ag price paths
will automatically hold for right-continuous price paths.
For example, the proof of Doob's fundamental Lemma~\ref{lem:Doob}
does not work for right-continuous price paths:
we will have $\lim_{n\to\infty}\tau_n(\omega)<\infty$
for some rational $a$ and $b>a$ whenever $\omega$ is not c\`adl\`ag.
(But there are ways around this difficulty,
as we saw in the proof of Corollary~\ref{cor:wild}.)

\begin{proposition}
  For typical $\omega\in R^+[0,T]$,
  $
    \vi(\omega)\le2
  $.
\end{proposition}

\begin{proof}
  Fix $p>2$ and set $\phi(u):=u^p$, $u\in[0,\infty)$.
  To see that Theorem~\ref{thm:main} continues to hold
  for the new sample space $\Omega=R^+[0,T]$,
  we will modify the proof of Proposition~\ref{prop:main}.

  Fix positive integer $L$.
  In view of Corollary~\ref{cor:wild},
  it suffices to construct a positive capital process
  that starts from a finite initial capital
  and attains final capital $\infty$ on all c\`adl\`ag $\omega$
  satisfying $\var_{\phi}(\omega)=\infty$ and $\sup\omega\le2^L$.
  Proceed as in the proof of Proposition~\ref{prop:main}
  until (\ref{eq:proof-1}),
  which should be replaced by
  \begin{equation*}
    S^{j,k}_T(\omega)
    \ge
    2^{-j}
    \MM_{k2^{-j}}^{(k+1)2^{-j}} (\omega)
    \wedge
    \frac{1}{w(j)};
  \end{equation*}
  the term $1/w(j)$ makes it possible to prevent the trading strategy
  leading to $S^{j,k}$
  from trading infinitely often.
  This will lead to
  \[
    S_{T}^j(\omega)
    \ge
    2^{-L-2j}
    \MM(\omega,2^{-j})
    \wedge
    \frac{1}{w(j)}
    \text{ when $\sup\omega\le2^L$}
  \]
  in place of (\ref{eq:proof-2}).
  In place of the first inequality in (\ref{eq:to-explain})
  we now have
  \[
    S_{T}(\omega)
    \ge
    \sum_{j=0}^{\infty}
    w(j)
    2^{-L-2j}
    \MM(\omega,2^{-j})
    \wedge
    1.
  \]
  In the case where
  $
    w(j)
    2^{-L-2j}
    \MM(\omega,2^{-j})
    >
    1
  $
  infinitely often our goal is achieved:
  $S_T(\omega)=\infty$.
  Therefore,
  we will assume that 
  $
    w(j)
    2^{-L-2j}
    \MM(\omega,2^{-j})
    \le
    1
  $
  for all $j\ge J$
  (where $J=J(\omega)$ depends on $\omega$).
  The chain (\ref{eq:to-explain})--(\ref{eq:end})
  can then be modified to
  \begin{align*}
    S_{T}(\omega)
    &\ge
    \sum_{j=J}^{\infty}
    w(j)
    2^{-L-2j}
    \MM(\omega,2^{-j})
    \ge
    \sum_{i\in I_{+}:j(i)\ge J}
    w(j(i))
    2^{-L-2j(i)}\\
    &=
    2^{-L}
    \sum_{i\in I_{+}:j(i)\ge J}
    \phi
    \left(
      2^{-j(i)}
    \right)\\
    &\ge
    2^{-L} 4^{-p}
    \sum_{i\in I_{+}:\omega(t_i)-\omega(t_{i-1})\le 2^{-J}}
    \phi
    \left(
      \omega(t_i)-\omega(t_{i-1})
    \right).
  \end{align*}
  Similarly,
  replacing the lower summation limit $j=0$ by $j=J$
  in the chain (\ref{eq:chain-start})--(\ref{eq:chain-end}),
  we obtain
  \begin{equation*}
    S_{T}(\omega)
    \ge
    2^{-L} 4^{-p}
    \sum_{i\in I_{-}:\omega(t_{i-1})-\omega(t_{i})\le 2^{-J}}
    \phi
    \left(
      \omega(t_{i-1})-\omega(t_{i})
    \right)
    -
    1.
  \end{equation*}
  Averaging the two lower bounds for $S_T(\omega)$
  now gives
  \begin{align*}
    S_T(\omega)
    &\ge
    2^{-L-1} 4^{-p}
    \sum_{i:\lvert\omega(t_{i})-\omega(t_{i-1})\rvert\le 2^{-J}}
    \phi
    \left(\left|
      \omega(t_{i})-\omega(t_{i-1})
    \right|\right)
    -
    \frac{1}{2}\\
    &\ge
    2^{-L-1} 4^{-p}
    \left(
      \sum_{i=1}^n
      \phi
      \left(\left|
        \omega(t_{i})-\omega(t_{i-1})
      \right|\right)
      -
      C(\omega)
    \right)
    -
    \frac{1}{2},
  \end{align*}
  where $C(\omega)<\infty$
  (by, e.g., \cite{Billingsley:1968}, Section~14, Lemma~1).
  Taking supremum over all partitions gives
  \begin{equation*}
    \left(
      \sup\omega\le2^L
      \And
      \var_{\phi}(\omega)=\infty
    \right)
    \Longrightarrow
    S_T(\omega) = \infty
  \end{equation*}
  in place of (\ref{eq:inequality}),
  which completes the proof.
\end{proof}

\section{Foundations}
\label{app:C}

We have considered three choices for the set of allowed price paths,
which we called the sample space:
$C^+[0,T]$ in Section~\ref{sec:continuous},
$D^{+}[0,T]$ in Section~\ref{sec:cadlag},
and $R^{+}[0,T]$ in Appendix~\ref{app:B}.
The assumption of continuity is traditional in this line of work
\cite{Takeuchi/etal:2009,\CTII,\CTIV},
and right-continuity is a natural relaxation of continuity
that agrees with the direction of time:
for each $t$, $\omega$ will not deviate much from $\omega(t)$ immediately after $t$.
The purpose of this appendix is to justify some details of our definitions,
and to discuss alternative definitions.

In the case of the sample spaces $D^+[0,T]$ and $R^+[0,T]$,
we defined $\FFF_t$ to be the universal completion of $\FFF^{\circ}_t$,
the $\sigma$-algebra generated by the projections
$\omega\mapsto\omega(s)$, $s\le t$.
In the case of the sample space $C^+[0,T]$
we could simply set $\FFF_t:=\FFF^{\circ}_t$
(as in \cite{\CTII,\CTIV}),
with the same definition of $\FFF^{\circ}_t$.
However, the most natural choice of $\FFF_t$
is to define it as the $\sigma$-algebra of all \emph{cylinder sets},
i.e., all sets $E\subseteq\Omega$ such that
\begin{equation}\label{eq:cylinder}
  \left(
    \omega\in E,
    \omega'\in\Omega,
    \omega|_{[0,t]} = \omega'|_{[0,t]}
  \right)
  \Longrightarrow
  \omega'\in E.
\end{equation}
The definitions of stopping times,
capital processes, upper probability, etc.,
stay the same, but they simplify greatly.
For example, a function is $\FFF_{\tau}$-measurable,
where $\tau$ is a stopping time,
if and only if it depends on $\omega$ only via its restriction
to the interval $[0,\tau(\omega)]$.

It turns out that the definitions based on the cylinder sets~\eqref{eq:cylinder} lead to a theory
that is in sharp contrast with our intuition about financial markets:
for example, in the case of the sample space $\Omega:=C^+[0,T]$,
the upper probability of a set $E\subseteq\Omega$ can take only two possible values:
it is 1 if $E$ contains a constant function and it is 0 otherwise.
See \cite{\CTXI} for a detailed discussion.
However, the following question remains open:

\begin{question}
  Let the sample space be the set $\Omega:=D^+[0,T]$ of all positive c\`adl\`ag functions
  $\omega:[0,T]\to[0,\infty)$.
  Let $E$ be the set of all $\omega\in\Omega$ satisfying $\vi(\omega)=2$.
  Is it true that $\UpProb(E)=1$ (or at least $\UpProb(E)>0$)?
\end{question}

At this point it is natural to show that we do not have similar problems
for the definitions of Sections \ref{sec:cadlag} and \ref{sec:continuous}
and Appendix \ref{app:B}.
For $t\in[0,T]$,
let $X_t:\Omega\to\bbbr$ be the projection $X_t(\omega):=\omega(t)$;
we will use this definition for
$\Omega:=D^+[0,T]$, $\Omega:=C^+[0,T]$, and $\Omega:=R^+[0,T]$.

\begin{proposition}\label{prop:coherence}
  Let $X_t$ be a martingale w.r.\ to a probability measure $P$
  on $(\Omega,\FFF_T)$
  and the filtration $(\FFF_t)$,
  where $\Omega$ is one of the spaces
  $\Omega:=D^+[0,T]$, $\Omega:=C^+[0,T]$, or $\Omega:=R^+[0,T]$.
  If $E\in\FFF_T$ satisfies $P(E)=1$,
  then $\UpProb(E)=1$.
\end{proposition}

\begin{proof}
  Under $P$,
  any positive simple capital process becomes a positive local martingale,
  since by the optional sampling theorem,
  every partial sum in (\ref{eq:simple-capital}) becomes a martingale.
  Every positive local martingale is a supermartingale%
  ,
  and so, by the monotone convergence theorem,
  any positive capital process (\ref{eq:positive-capital})
  is a positive supermartingale
  (not necessarily right-continuous).
  Therefore,
  the existence of a positive capital process $S$
  increasing its value between times $0$ and $T$
  by more than a strictly positive constant for all $\omega\in\Omega$
  would contradict $\int S_T\dd P\le\int S_0\dd P$.
\end{proof}

\noindent
Proposition~\ref{prop:coherence} shows
that the results of Sections \ref{sec:cadlag} and \ref{sec:continuous}
and Appendix \ref{app:B} 
have many implications for typical paths of numerous stochastic processes,
including Brownian motion,
which is continuous and has typical paths $\omega$ satisfying $\vi(\omega)=2$.

\section{Details of the proof of Proposition \ref{prop:coherence}}
\label{app:D}

The proof of Proposition~\ref{prop:coherence} relies on the fact
that each addend in (\ref{eq:simple-capital})
(and, therefore, each partial sum in (\ref{eq:simple-capital}))
is a martingale when $\omega$ is a martingale.
In this appendix we will check carefully this property.
The argument is obvious,
but it might be useful to spell it out.

We know that, even when $\Omega=R^+[0,T]$,
almost all $\omega\in\Omega$ are c\`adl\`ag
(\cite{Dellacherie/Meyer:1982}, Theorem VI.3),
which allows us to apply the optional sampling theorem
(see, e.g., \cite{Revuz/Yor:1999}, Theorem II.3.2).

Each addend in (\ref{eq:simple-capital}) can be rewritten as
\begin{equation*}
  h_n(\omega)
  \bigl(
    \omega(\tau_{n+1}\wedge t)-\omega(\tau_n\wedge t)
  \bigr)
  =
  h_n(\omega)
  \bigl(
    \omega(t)-\omega(\tau_n\wedge t)
  \bigr)
  -
  h_n(\omega)
  \bigl(
    \omega(t)-\omega(\tau_{n+1}\wedge t)
  \bigr),
\end{equation*}
and so it suffices to prove that
\begin{equation}\label{eq:addend}
  h'_n(\omega)
  \bigl(
    \omega(t)-\omega(\tau_n\wedge t)
  \bigr),
\end{equation}
where $h'_n$ is bounded and $\FFF_{\tau_n}$-measurable,
is a martingale.
For each $t\in[0,\infty)$,
(\ref{eq:addend}) is integrable
by the boundedness of $h'_n$ and the optional sampling theorem.
We only need to prove, for $0<s<t$, that
(omitting, until the end of the proof,
the prime in $h'$, the argument $\omega$, and ``a.s.'')
\begin{equation*} 
  \Expect
  \left(
    h_n
    \bigl(
      \omega(t)-\omega(\tau_n\wedge t)
    \bigr)
    \mid
    \FFF_s
  \right)
  =
  h_n
  \bigl(
    \omega(s)-\omega(\tau_n\wedge s)
  \bigr).
\end{equation*}
We will check this equality on two $\FFF_s$-measurable events separately:
\begin{description}
\item[$\{\tau_n\le s\}$:]
  We need to check
  \begin{equation*}
    \Expect
    \left(
      h_n
      \bigl(
        \omega(t)-\omega(\tau_n)
      \bigr)
      \III_{\{\tau_n\le s\}}
      \mid
      \FFF_s
    \right)
    =
    h_n
    \bigl(
      \omega(s)-\omega(\tau_n)
    \bigr)
    \III_{\{\tau_n\le s\}}.
  \end{equation*}
  Since $h_n\III_{\{\tau_n\le s\}}$ is bounded and $\FFF_s$-measurable
  (its $\FFF_s$-measurability follows, e.g.,
  from Lemma 1.2.15 in \cite{Karatzas/Shreve:1991} and the monotone-class theorem),
  it suffices to check
  \begin{equation*}
    \Expect
    \left(
      \bigl(
        \omega(t)-\omega(\tau_n)
      \bigr)
      \III_{\{\tau_n\le s\}}
      \mid
      \FFF_s
    \right)\\
    =
    \bigl(
      \omega(s)-\omega(\tau_n)
    \bigr)
    \III_{\{\tau_n\le s\}}.
  \end{equation*}
  Since $\omega(\tau_n)\III_{\{\tau_n\le s\}}$ is $\FFF_s$-measurable,
  it suffices to check
  \begin{equation*}
    \Expect
    \left(
      \omega(t)
      \III_{\{\tau_n\le s\}}
      \mid
      \FFF_s
    \right)
    =
    \omega(s)
    \III_{\{\tau_n\le s\}}.
  \end{equation*}
  The stronger equality
  $
    \Expect
    \left(
      \omega(t)
      \mid
      \FFF_s
    \right)
    =
    \omega(s)
  $
  is part of the definition of a martingale.
\item[$\{s<\tau_n\}$:]
  We are required to prove
  \begin{equation*}
    \Expect
    \left(
      h_n
      \bigl(
        \omega(t)-\omega(\tau_n\wedge t)
      \bigr)
      \III_{\{s<\tau_n\}}
      \mid
      \FFF_s
    \right)
    =
    0,
  \end{equation*}
  but we will prove more:
  \begin{equation*}
    \Expect
    \left(
      h_n
      \bigl(
        \omega(t)-\omega(\tau_n\wedge t)
      \bigr)
      \III_{\{s<\tau_n\}}
      \mid
      \FFF_{s\vee\tau_n\wedge t}
    \right)
    =
    0
  \end{equation*}
  ($s\vee x\wedge t$ being a shorthand
  for $(s\vee x)\wedge t$ or, equivalently, $s\vee(x\wedge t)$).
  Since the event $\{\tau_n\le t\}$,
  being equal to $\{\tau_n\le s\vee\tau_n\wedge t\}$,
  is $\FFF_{s\vee\tau_n\wedge t}$-measurable
  (see \cite{Karatzas/Shreve:1991}, Lemma~1.2.16),
  it is sufficient to prove
  \begin{equation}\label{eq:remains}
    \Expect
    \left(
      h_n
      \bigl(
        \omega(t)-\omega(\tau_n\wedge t)
      \bigr)
      \III_{\{s<\tau_n\le t\}}
      \mid
      \FFF_{s\vee\tau_n\wedge t}
    \right)
    =
    0
  \end{equation}
  and
  \begin{equation*}
    \Expect
    \left(
      h_n
      \bigl(
        \omega(t)-\omega(\tau_n\wedge t)
      \bigr)
      \III_{\{t<\tau_n\}}
      \mid
      \FFF_{s\vee\tau_n\wedge t}
    \right)
    =
    0.
  \end{equation*}
  The second equality is obvious,
  so our task has reduced to proving the first, (\ref{eq:remains}).
  Since $h_n\III_{\{\tau_n\le t\}}=h_n\III_{\{\tau_n\le s\vee\tau_n\wedge t\}}$
  is bounded and $\FFF_{s\vee\tau_n\wedge t}$-measurable,
  (\ref{eq:remains}) reduces to
  \begin{equation*}
    \Expect
    \left(
      \bigl(
        \omega(t)-\omega(\tau_n\wedge t)
      \bigr)
      \III_{\{s<\tau_n\le t\}}
      \mid
      \FFF_{s\vee\tau_n\wedge t}
    \right)
    =
    0,
  \end{equation*}
  which is the same thing as
  \begin{equation*}
    \Expect
    \left(
      \bigl(
        \omega(t)-\omega(s\vee\tau_n\wedge t)
      \bigr)
      \III_{\{s<\tau_n\le t\}}
      \mid
      \FFF_{s\vee\tau_n\wedge t}
    \right)
    =
    0.
  \end{equation*}
  The last equality follows from the $\FFF_{s\vee\tau_n\wedge t}$-measurability
  of the event
  \[
    \{s<\tau_n\le t\}
    =
    \{s<s\vee\tau_n\wedge t\}
    \cap
    \{\tau_n\le s\vee\tau_n\wedge t\}
  \]
  (see \cite{Karatzas/Shreve:1991}, Lemma 1.2.16)
  and the special case
  \begin{multline*}
    \Expect
    \left(
      \bigl(
        \omega(t)-\omega(s\vee\tau_n\wedge t)
      \bigr)
      \mid
      \FFF_{s\vee\tau_n\wedge t}
    \right)\\
    =
    \omega(s\vee\tau_n\wedge t) - \omega(s\vee\tau_n\wedge t)
    =
    0
  \end{multline*}
  of the optional sampling theorem.
\end{description}

\end{document}